\documentclass[a4paper,USenglish,cleveref, autoref, thm-restate]{lipics-v2021}

\nolinenumbers
\usepackage{graphicx} 
\usepackage{algorithm}
\usepackage{algorithmic}
\usepackage{comment}
\usepackage{xcolor, soul}
\usepackage{amsmath,amsthm,amssymb}
\usepackage{mathtools}

\newtheorem{thm}{Theorem}
\newtheorem{rmk}{Remark}
\newtheorem{fact}{Fact}

\newcommand{\world}[1]{\mathbf{E}_ #1}
\renewcommand{\Pr}[1]{\mathbb{P}\left[   #1\right]}
\newcommand{\E}[1]{\mathbb{E}\left[ #1 \right]}

\newcommand{\Lrwf}{\frac{\log(1/\delta_0 - 1)}{\log(p/(1-p))}}
\newcommand{\Lrw}{L}
\newcommand{\alpp}{\left(\frac{1-p}{p}\right)}
\newcommand{\pmin}{p_\text{min}}

\title{Pessimism Traps and Algorithmic Interventions}
\titlerunning{Pessimism Traps and Algorithmic Interventions}
\author{Avrim Blum}{Toyota Technological Institute at Chicago, USA}{avrim@ttic.edu}{}{}
\author{Emily Diana}{Carnegie Mellon University, Tepper School of Business, USA}{ediana@andrew.cmu.edu}{}{}
\author{Kavya Ravichandran}{Toyota Technological Institute at Chicago, USA}{kavya@ttic.edu}{https://orcid.org/0000-0002-4133-4168}{}
\author{Alexander Tolbert}{Emory University, Departnment of Quantitative Theory and Methods, USA}{alexander.tolbert@emory.edu}{}{}
\authorrunning{A. Blum, E. Diana, K. Ravichandran, A. Tolbert}

\Copyright{Avrim Blum, Emily Diana, Kavya Ravichandran, and Alexander Tolbert} 


\relatedversion{This is the full version of the paper with the same title published at FORC 2025}

\funding{This work was supported in part by the National Science Foundation under grants CCF-2212968 and ECCS-2216899, by the Simons Foundation under the Simons Collaboration on the Theory of Algorithmic Fairness, and by the Office of Naval Research MURI Grant N000142412742.
}
\EventEditors{Mark Bun}
\EventNoEds{1}
\EventLongTitle{6th Symposium on Foundations of Responsible Computing (FORC 2025)}
\EventShortTitle{FORC 2025}
\EventAcronym{FORC}
\EventYear{2025}
\EventDate{June 4--6, 2025}
\EventLocation{Stanford University, CA, USA}
\EventLogo{}
\SeriesVolume{329}
\ArticleNo{5}
\begin{document}

\maketitle

\begin{abstract}
In this paper, we relate the philosophical literature on pessimism traps to information cascades, a formal model derived from the economics and mathematics literature. A pessimism trap is a social pattern in which individuals in a community, in situations of uncertainty, copy the sub-optimal actions of others, despite their individual beliefs. This maps nicely onto the concept of an information cascade, which involves a sequence of agents making a decision between two alternatives, with a private signal of the superior alternative and a public history of others' actions. Key results from the economics literature show that information cascades occur with probability one in many contexts, and depending on the strength of the signal, populations can fall into the incorrect cascade very easily and quickly. Once formed, in the absence of external perturbation, a cascade cannot be broken -- therefore, we derive an intervention that can be used to nudge a population from an incorrect to a correct cascade and, importantly, maintain the cascade once the subsidy is discontinued. We extend this to the case of multiple communities, each of which might have a different optimal action, and a government providing subsidies that cannot discriminate between communities and does not know which action is optimal for each.  We study this both theoretically and empirically.
\end{abstract}

\section{Introduction}

Studying decision-making under uncertainty has interested scholars in several disciplines, such as psychology, economics, and philosophy. In economics, specifically in statistical discrimination, scholars such as \cite{coateLoury} have studied how individuals make decisions, particularly in marginalized communities where societal structures and perceived barriers can impact their choices. This often leads to self-censorship or a withdrawal from available opportunities due to a pessimistic outlook. The concept of pessimism has also been explored in philosophy, specifically in recent work by \cite{morton2022resisting}. \cite{morton2022resisting}'s research examines the dilemmas faced by individuals belonging to marginalized groups when presented with evidence that diminishes their chances of success. As a result, they may rationally choose to redirect their efforts elsewhere to pursue other less valuable alternatives. However, this decision only reinforces the pessimistic mindset that influenced it, perpetuating a cycle of limited aspirations and self-doubt. \cite{morton2022resisting} criticizes the simplistic solution of turning towards optimism and proposes a more critical and situational application of optimistic beliefs.
\cite{morton2022resisting} characterizes \textit{pessimism traps} as follows: 
\begin{quote}
A \textbf{pessimism trap} \dots is
meant to capture how negative beliefs about one’s likelihood of success can play a role in agents pursuing less risky, modest ends instead of ambitious ones, thereby further entrenching the negative evidence that the agent herself and other agents in a similar position face. The central idea, however, has broader applicability.\dots
for this paper, I will focus on those cases in which poverty, prejudice, and discrimination play a role in providing agents with the sort of evidence that would \textit{prima facie} make it \textbf{rational} for them to arrive at the pessimistic beliefs that play a role in thwarting their ambition.
\end{quote}

Our paper extends Morton's qualitative description of pessimism traps (summarized in Appendix~\ref{appendix:features}) with mathematical formalism and studies algorithmic decision-making within this formalism.  We explore the theoretical model of \textit{information cascades} (also referred to as ``herding'' in the opinion dynamics literature \cite{mossel1, pathological}) to shed light on decision-making under uncertainty and pessimism in such contexts. Using this model, we propose interventions to break pessimism traps and redirect the population into favorable states.

\subsection{Our Contributions}

In this paper, we study pessimism traps as conceptualized by Morton and offer interventions to shift communities out of these traps. Our first major contribution is to link the concept of a pessimism trap to the opinion dynamics literature. We do so by interpreting a single-dimensional information cascade model as a decision sequence in a community. We extend this in Section~\ref{sec:k-group} to consider $k$ parallel information cascades representing the decision processes of $k$ independent communities.

Our next major contribution is an intervention for the single community setting that sustainably shifts communities out of pessimism. We consider a subsidy enacted by the government to incentivize certain behavior in an agent. Two natural ideas fail: first, suppose we subsidize the ambitious action by a large fixed amount. While this intervention incentivizes agents to act ambitiously in the short term, once the subsidy program ends, subsequent agents will ascribe prior agents' ambition entirely to the subsidy they received and fall back into pessimism traps. Secondly, subsidizing by too small an amount fails to incentivize the ambitious action over the moderate one at all. 
Thus, we derive the precise size of the subsidy that will incentivize those who were already leaning toward ambition while not moving those who were anyway not considering the ambitious choice. Herein lies our first surprising insight -- the non-monotonicity associated with the effect of the size of subsidy.

Next, we study $k$ communities, each behaving as above. The government still intervenes but must act impartially with respect to community membership. 
Each of these groups may have a different optimal choice among the two options, and the government (a) does not know which choice is optimal for each group and (b) must be blind to  community membership in providing subsidies. We show that in this case, we can construct a distribution $\mathcal{D}$ such that if the government draws the subsidy randomly from that distribution, then eventually each community will shift toward what is the optimal end for them. This intervention relies crucially on the previous insight: if the provided subsidy is in the ``just right'' range discussed above for an agent, the agent's action will shift their community toward optimal action. However, if it is too large or too small (outside of that range), the community will remain in the same state. Thus, provided the distribution assigns at least a minimum probability to the ``just right'' range for a fixed community, that group will eventually settle into their optimal action. Therefore, our second main insight is that since we can guarantee (1) the existence of a range of subsidy values that shift the community toward the correct action for that community and (2) a single community faces no {\em negative} effect when provided a subsidy outside of that range, the government can use randomization to shift {\em all} groups into optimal behaviors, even without knowing what is optimal for any given community.

Finally, we verify the effectiveness of our proposed intervention in the one-group case via experiments. Our experiments show that the intervention we develop is indeed successful in shifting sequential decisions into optimism. Further, we confirm that the required budget for the intervention is tractable. Our experiments provide additional confirmation of the benefits of our intervention in the simple model that we study.

We view our work as a step forward in providing a theoretical model in which to study pessimism traps and, more broadly, in developing empowering interventions for marginalized communities. While our 
model is necessarily stylized, our major insights could help guide interventions in more complex settings. 

In summary, our contributions are:
\begin{enumerate}
    \item Linking the opinion dynamics and pessimism traps literatures by studying pessimism trap formation in the information cascade model.
    \item Identifying an intervention that sustainably shifts a community of agents away from pessimism traps. 
    \item Extending this result into a setting where there are $k$ different groups and where the government knows neither which the correct action is for a group, nor how close they are to escaping the trap and showing the power of randomization in this setting.
    \item Corroborating the success of our theoretical interventions experimentally.
\end{enumerate}

The remainder of our paper is organized as follows. First, we briefly survey work from a diverse array of fields related to our work.
Next, we present the formal mathematical model in which we study pessimism traps along with some preliminary facts about them, and we reflect on the strengths and weaknesses of this modeling approach. In Section~\ref{sec:subsidy}, we introduce and analyze our intervention to shift a single group toward optimism. We then extend this in Section~\ref{sec:k-group} to the setting with $k$ groups each of whom has an unknown optimal action.  Finally, we provide experiments that show the success of our proposed interventions, substantiating our theoretical results.

\subsection{Related Work}

Here, we provide a brief survey of related works that contextualize our approach to studying pessimism traps and interventions aimed at promoting optimism.  

As quoted above, \cite{morton2022resisting} formalizes the notion of pessimism traps, building upon several empirical characterizations from scholars studying the phenomenon in the field of education \cite{cohen_2006_intervention}. Morton highlights the importance of {\em belief}, which is why we focus our modeling and interventions on shifting beliefs in a sustainable manner.

Much work in the literature studies the connection between individual behavior and community-wide effects. This includes the literature on opinion dynamics and social learning, which are surveyed in \cite{sirbu_opinion_2017} and \cite{chamley_rational_2004}. They discuss both simple and more complex models for how people's beliefs vary in relation to the beliefs of those with whom they interact. \cite{acemoglu_opinion_2011} studies both Bayesian and non-Bayesian models for how agents update their beliefs, investigating consensus and asymptotic learning of state (i.e., what are the true beliefs in the world). On the behavioral economics side, \cite{kahneman1974judgment} study the psychological reasons behind why individuals might act in a way that disregard their private information.

Since individual behavior can impact community-wide outcomes, as studied in the works above, it is reasonable to consider ``nudging'' or intervening at the individual level. Works that have studied this include \cite{thaler2008nudge} in behavioral economics, showing how minor policy adjustments can realign individual decision-making with optimal outcomes.
 Similarly, work in theoretical computer science and game theory, such as \cite{Balcan2013CircumventingTP}, study ``nudging'' specifically in the context of equilibria, where the goal is to redirect a population from a less desirable to a more desirable equilibrium.
 
Additionally, empirical studies have supported the application of nudges in various domains. For example, \cite{johnson2012beyond} extends the discussion on the efficacy of nudges in real-world settings. Integrating nudging into public policy, particularly in health and environmental strategies, as discussed by \cite{benartzi2017should}, demonstrates how these concepts have evolved beyond theoretical discussions to practical implementations. 

Therefore, we utilize information cascade models to develop ``nudging''-style interventions aimed at combating the pessimism traps that frequently arise in marginalized communities. 
 
\cite{bikhchandani1992theory} initially framed the classic information cascade model, demonstrating how individuals, despite possessing private information, often conform to the erroneous actions of predecessors due to the strong influence of prior actions. This model has served as a baseline for exploring various dimensions of information processing within groups. Building on this foundational work,  \cite{anderson1997information} conducted laboratory experiments to observe cascade behavior in controlled settings, adding empirical evidence to the theoretical predictions. We mainly focus on the former.

\section{Preliminaries and Basic Model}
\label{sec:prelims}
We study two settings in this paper: first, we look at a single sequence of $T$ agents deciding whether to take action  $A$ or $B$; second, we will expand this to $k$ parallel sequences of agents, where the best action (between $A$ and $B$) differs group-to-group. Here, we define the signal and sequence dynamics for a single sequence; in Section~\ref{sec:k-group}, we flesh out the setting with $k$ parallel groups. We take $T$ to be finite, but our analysis extends to infinite $T$. 
\subsection{Single Sequence Model}
First, for the case of a single community (sequence) 
we assume that $A$ is uniformly the better option (that is, it is best for all agents), though this fact is unknown to the agents. When we consider $k$ parallel sequences, then each sequence will have its own better option. We let $\world{A}$ represent the event that action $A$ is correct and $\world{B}$ represent the event that action $B$ is correct. When referring to the set of these possible events, we define $\mathcal{E} \triangleq \{\world{A}, \world{B}\}$ to indicate the set of possibilities. In particular, we assume that \textit{a priori}, agents have no bias toward either action, which we model by saying that they have a common, uniform prior over $\mathcal{E}$. Formally, $\Pr{\world{A}}=\frac{1}{2}$ and $\Pr{\world{B}} = \frac{1}{2}$. Further, one action has associated with it reward 1 and the other reward 0. 
In addition to the public prior, each agent $t$ receives a private signal $s_t$ indicating that one of $A$ or $B$ is the correct action.  This signal is correct with probability $ p > \frac{1}{2}$ and incorrect with probability $1-p$. Therefore, we have $\Pr{s_t|\world{A}} = p^{\delta_{s_t,A}} (1-p)^{\delta_{s_i,B}}$ and $\Pr{s_t|\world{B}} = p^{\delta_{s_t,B}} (1-p)^{\delta_{s_i,A}}$, where \( \delta \) is the Kronecker delta function.
\newcommand{\barH}[1]{\bar{H}_{#1}}

Given this signal and the actions of preceding agents, the agent decides whether to take action $A$ or $B$. Let $h_t$ indicate the action of agent $t$, where an agent is identified by the time they act and $H_{t-1}$ be the history of actions for the first $t-1$ agents. Agent $t$ will take an action if its expected reward exceeds the expected reward of the other action.  For concreteness, we adopt the tie-breaking convention that if an agent is indifferent between the two options based on the calculated posteriors, they follow their private signal. Now, we define an information cascade.
\begin{definition}[Information Cascade]
An information cascade occurs when an agent's action does not depend on their private signal. This means that one of the following occurs:

\begin{enumerate}
\item $\Pr{\world{A}|s_t=1,H_{t-1}} > \frac{1}{2} \text{ and } \Pr{\world{A}|s_t=0,H_{t-1}} > \frac{1}{2}$

\item $\Pr{\world{A}|s_t=1,H_{t-1}} < \frac{1}{2} \text{ and }\Pr{\world{A}|s_t=0,H_{t-1}} < \frac{1}{2}$
\end{enumerate}

\end{definition}
\begin{rmk}
A cascade begins when the observed history \( H_{t-1} \) becomes sufficiently skewed towards either adoption or rejection. Therefore, enough previous agents have taken the same action to make it rational for subsequent agents to follow the trend, based on a simple Bayesian updating procedure, even if the history contradicts their private signal.
\end{rmk}

To formally capture the concept of a pessimism trap for a sequence setting, we define the conditions under which agents, influenced by the actions of their predecessors, consistently choose the inferior option $B$. This phenomenon occurs when the aggregated evidence from prior decisions leads to a bias that overrides the agents' private signals. We represent this situation mathematically as follows:

\begin{definition}[Pessimism Trap]
A pessimism trap occurs when an information cascade leads agents to consistently choose the inferior action $B$ due to an overwhelming influence of prior agents' incorrect actions. Formally, a pessimism trap is defined by the following condition:
$$\Pr{\world{B}|s_t=1, H_{t-1}} > \frac{1}{2} \text{ and }\Pr{\world{B}|s_t=0, H_{t-1}} > \frac{1}{2}$$

This indicates that despite the private signal $s_t$ suggesting action $A$ is better, history $H_{t-1}$ has led agent $t$ to believe that $B$ is better. Subsequently, agents' decisions are based purely on observed history rather than private signals. 
\end{definition}

\begin{rmk}
When computing posteriors, we assume that agents only consider actions of those up until a cascade began, since if agents are rational, they realize that no further inferential information can be obtained from individuals who do not use their signals. 
\end{rmk}

As a result, we will often be interested in $\barH{t-1},$ the history of actions taken by agents who chose an action if and only if it was their signal. We express mathematically the Bayesian updating procedure through which agents come to their posterior belief about the best action based on the history of actions before a cascade begins and then formalize the kinds of governmental / central interventions we study:
\begin{equation}
\label{eq:post}
\Pr{\world{A}|s_t, H_{t-1}} = \Pr{\world{A}|s_t, \barH{t-1}} = \frac{\Pr{s_t|\world{A}}\Pr{\barH{t-1}|\world{A}}}{\sum_{w \in \{A,B\}} \Pr{s_t|\world{w}} \Pr{\barH{t-1}|\world{w}}}
\end{equation}

\begin{definition}
    A subsidy of size $r$ toward action $a$ is a benefit provided to an agent taking said action independent of the true world. Thus, in the world where $a$ is the correct action, instead of receiving just $R$ reward for taking action $a\,,$ the agent receives $R+r\,$ reward, and in a world in which it is the incorrect action, the agent still receives $r$ reward.
\end{definition}

\textbf{Formulation as Random Walk.} The main technical tools used for analysis of the cascade and subsidy come from framing the process as a random walk. We isolate the probability with which action $A$ is taken and the probability with which action $B$ is taken, and we analyze the information cascade as taking steps on a random walk. This walk in general terminates when either an up or a down cascade has been reached. Once we introduce a subsidy, the stopping point of interest will be the up cascade.
In Appendix~\ref{appendix:random-walk}, we show how we can use the random walk formalism to derive the length of time for which the subsidy must be in place and the required budget.

\subsection{On Our Modeling Choices}
A significant contribution of our work is that we propose using the information cascade model to study pessimism traps and associated interventions. In this section, we justify our use of this model and discuss the trade-offs made.

\textbf{Justification:}
A natural question is: why choose this particular level of abstraction, i.e., why model the pessimism trap as hinging in this way on decisions made by previous agents? Several works have empirically characterized this phenomenon in the context of education, building on which Morton gives an epistemic characterization. At its core, as conceptualized by Morton, a pessimism trap occurs as a result of people's beliefs (about the world, about the rationality of others, etc), realized as herding behavior. Thus, in order to model this faithfully, we require a model in which we may quantitatively update the {\em beliefs}. A Bayesian formulation is natural for this. Further, though sequentialness is not inherent to Morton's characterization of pessimism traps, there is a natural sense in which agents consider the context of those who went {\em before} them when making decisions. Accordingly, a sequential model is a good choice for organizing the manner in which agents are influenced.

\textbf{Strengths:} In some sense, the Bayesian posterior update represents the ``optimal'' rational decision. By modeling pessimism traps this way, we are showing that even with perfect rationality, a trap can form, corroborating Morton's point that pessimism traps do not occur due to lack of rationality. Further, as summarized in the appendix, two key features Morton identifies regarding pessimism traps are: (1) there is evidence of similar people not succeeding at the ambitious end and (2) not pursuing the ambitious end will not change the agent's view of its value. Both of these are well-represented in the information cascade model. Particularly once we extend to multiple groups, agents are looking to the history of actions taken by people like them (i.e., in the same group), and witnessing several people taking a certain action would indicate to them that taking the opposite action does not tend to be beneficial for people like them. Likewise, in the information cascade model, there is no feedback for whether a different action would have been correct, and therefore there is no reason for an agent to change their pessimistic view about it.

\textbf{What this model misses:}
On the other hand, since the reward from the taken action is received in one step, this model does not reflect the fact that the ambitious choice requires investment and is contrasted with a choice that has a reasonable payoff throughout. Similarly, the model does not capture risk associated with the ambitious choice. 
These are important modeling considerations for future work.

\section{Time-Varying Subsidy} \label{sec:subsidy}
Without external intervention, once a cascade begins, it persists by definition. However, is it possible to derive an intervention from an external entity, such as the government, that can lift a population out of an incorrect cascade and redirect it toward a correct one? Importantly, can this subsidy be designed so that the correct cascade remains stable once the subsidy is removed? Note that the na\"{i}ve strategy of subsidising the ``correct'' action fails the sustainability criterion, because upon removal, agents have no reason to believe that that action was correct. They would simply believe that agents who took it during the subsidized period did so due to said subsidy. In this section, we focus on the question of how to design a good intervention for a single group / sequence that has {\em sustainable} effects. 

Since a cascade occurs when an agent makes the same decision regardless of their signal, breaking a cascade requires influencing at least some agents to act according to their respective private signals. Indeed, this is our approach to designing a subsidy.
We consider subsidies for the ``correct'' action, which, in this section, we assume to be action $A$. The net reward for choosing the correct action is $R$.

Let the subsidy the government provides at time $t$ be $r_t$. Recall the two possible states of the world: world $A$, where $A$ is the correct action, and world $B$, where $B$ is the correct action. Let $|A|$ indicate the number of choices for $A$ outside of a cascade state, and similarly for $|B|$. Finally, we assume that the entity providing this subsidy is not trusted by the agents, and so the agents cannot infer from the direction of the subsidy which action is correct. 

Algorithm \ref{alg:subsidy} implements this subsidy scheme. We assume that this algorithm is applied after an incorrect cascade has already begun, and the purpose is to strategically provide a subsidy $r_t$ to the agent acting at time $t$ to break the community out of the pessimism trap. When agents are acting according to their own signals, or when the correct cascade has been reached, the subsidy need no longer be applied. The guarantees and derivations for Algorithm \ref{alg:subsidy} are provided in Theorems \ref{thm:subsidy} and~\ref{thm:subsidy_budget}, with complete proofs in the appendix. The main idea is that once agents act in accordance with their signals, after we see enough agents reveal their signal, a simple majority vote will, in expectation, reveal the correct action. At this point, even if the subsidy is removed, rational agents will act optimistically.
\begin{algorithm}
\caption{Redirecting Pessimism Traps}
\begin{algorithmic}[1]
\REQUIRE Start of incorrect cascade $t'$, history $H_{t''}$ for $t'' \geq t'$, $(T - t'')>> \frac{4}{2p-1}$, correct action $A$, incorrect action $B$, private signals $\{s\}_{t=1}^{T}$, signal strength $p$, reward $R$ for correct action
\STATE $|A| \gets$ choices for $A$ in $\barH{t'-1}$
\STATE $|B| \gets$ choices for $B$ in $\barH{t'-1}$
\STATE $ t \gets t''$
\WHILE{$t < T$}
\IF{$|A| - |B| \leq - 2$ (In incorrect cascade)}
\STATE $\gamma_t \gets \left( \frac{1-p}{p} \right)^{2|A|-t} \frac{1-p}{p}$ ; $r_t \gets R \left( \frac{\gamma_t -1 }{1 + \frac{1- p}{p} \gamma_t}\right)$
\ELSE
\STATE $r_t \gets 0$ 
\ENDIF
\IF{$r_t + R \cdot \Pr{\world{A} \mid H_{t-1}, s_t} > R \cdot \Pr{\world{B} \mid H_{t-1}, s_t}$}
\STATE Agent chooses action $A$
\IF{$|A| - |B| \leq 1$ (Not in cascade)}
\STATE $|A| = |A| + 1$
\ENDIF
\ELSE
\STATE Agent chooses action $B$
\IF{$|A| - |B| \leq 1$ (Not in cascade)}
\STATE $|B| = |B| + 1$
\ENDIF
\ENDIF
\STATE Update history and increment $t$
\ENDWHILE
\end{algorithmic}
\label{alg:subsidy}
\end{algorithm}

\begin{thm}
\label{thm:subsidy}
The subsidy scheme used in Algorithm 1 causes all agents to act according to their signals until the population falls into the correct cascade. The subsidy value is
$
r_t = R \left( \frac{\gamma_t -1 }{1 + \frac{1- p}{p} \gamma_t}\right)\,,
$
where $\gamma_t = \left( \frac{1-p}{p} \right)^{2|A|-t} \frac{1-p}{p}\,.$
\end{thm}

\textbf{Proof Sketch} We want the subsidy to incentivize taking action $A$ only if it is already aligned with the agent's signal.
We seek to determine $r_t$ such that it ensures action $A$ is at least as preferable when $S_t = A$ and less preferable when $S_t = B$. That is, we want $r_t$ to satisfy the following conditions, so we simply solve:
\begin{align*} \text{(1) } r_t + &R \cdot \Pr{\world{A} \mid H_{t-1}, s_t = A} \geq R \cdot \Pr{\world{B} \mid H_{t-1}, s_t = A} \\ \text{(2) } r_t + &R \cdot \Pr{\world{A} \mid H_{t-1}, s_t = B} \leq R \cdot \Pr{\world{B} \mid H_{t-1}, s_t = B}
\end{align*}

Due to the fact that the subsidy induces signal-revealing, in order to compute the posterior, we can view the history as a series of revealed signals and consider the likelihood of seeing that stream in each of the worlds. $\hfill \square$

\begin{thm}
\label{thm:subsidy_budget}
In expectation, the subsidy is provided over fewer than $\frac{4}{2p-1}$ rounds and totals no more than $R \frac{4}{2p-1}$.
\end{thm}
\textbf{Proof Sketch} To show this, we appeal to the random walk formulation discussed above. If the subsidy value is too small, the agent will decide based on the history, if the subsidy value is too large, the agent will choose the subsidized action, and if the subsidy is ``just right,'' the agent will act according to their signal. The first two cases correspond to not taking any steps on the random walk, and the last case corresponds to taking a step in the direction of the signal on the random walk. Thus, we can analyze the number of total steps that need to be taken in order to net sufficiently many steps to the correct directions. We apply Wald's equation to do this. 
A detailed computation can be found in the Appendix ~\ref{appendix:proof-subsidylen}.  $\hfill \square$

\section{Extension to multiple groups} \label{sec:k-group}
In this section, we remove the strong assumptions of the previous section, namely that (1) the government knows the correct action and (2) everyone has the same correct action.
    Perhaps a more natural assumption would be that there are $k$ groups, where each group has its own correct action (between the two, A and B). 
    In this section, we show how to develop a subsidy scheme in which {\em all groups} end up in their respective correct cascade, even if the government is not privy to the respective correct actions.
    
    In this setting, an agent in a group only sees other agents from their group. For example, consider a small town where the decisions of a student are primarily, if not completely, affected by those in the same town. 
    Suppose the government knows the strength of signals provided to agents, but it does not know the agent's town nor which action they take. At each time $t$, the government may give a subsidy $r_t$ with the goal that, eventually, all groups will reach a stable up cascade on the correct action for that group.

    In this section, we show how the government can construct a distribution $\mathcal{D}$ without knowledge of anything more than described above, such that if at each time $r_t \sim \mathcal{D}\,,$ after sufficiently many steps, with high probability, all groups will have converged to their respective correct cascades.
    Importantly, the key idea here is the same as in Algorithm \ref{alg:subsidy}: one kind of subsidy that will help each group make decisions that are right for them in the long run is one that encourages an agent to reveal their signal.
    Further, recall that we derived exactly a signal-revealing subsidy in the previous section. Thus, we know that for each value of $|A|-|B|\,,$ i.e., for each location along the random walk, there exists a value of the subsidy that incentivizes the agent to reveal their signal. Provided this subsidy value is chosen with at least some fixed minimum probability, we can use a similar random walk analysis to before, only this time slowed down by that probability factor.
    In this section, we flesh this argument out. 
    Proofs can be found in Appendix~\ref{appendix:k-group}.

\subsection{Formally Defining The Setting}
We define the setting formally below.

\begin{definition} \label{defn:multigroup-game}
    Suppose that in the world, there are two potential actions $A$ and $B\,.$ Each agent in the world has an index in $[k]$ associated with them which we call their group. For each group in $[k]\,$, one of the actions is the ``right'' one and the other is the ``wrong'' one. Notably, which action is correct differs between groups. Formally, we may assume the correct action for a group is $A$ with probability 1/2 and $B$ with probability 1/2. \par
    At time $t\,,$ the ``universe'' draws an index $j \sim \mathcal{G}\,,$ where $\mathcal{G}$ is a distribution supported on $[k]\,,$ with $j$ signifying the group index, and the minimum value of the probability mass function is $g_\text{min}$. It is a new agent's turn to make a decision, and they are a member of group $j\,.$ This agent sees the subsidy history for all agents and need rationally only consider the action history of those in group $j$ who have gone before them. They receive a private signal $s_t\,,$ which is the correct action for the group to which they belong with probability $p_j > 1/2$ and incorrect with probability $1-p_j\,.$ They consider their group history and private signal to update their posterior belief as shown in Equation~\ref{eq:post}, and then incorporate the present subsidy to make their final decision.
\end{definition}

Next, we define the subsidy scheme followed by the government. In fact, it suffices to define a uniform distribution over the possible values the subsidy would need to take, i.e., all values that would encourage signal-revealing in an agent before that agent's group hits a cascade. 

\begin{definition} \label{defn:government-subsidy}
    Let $v_{x, p}$ be the size of subsidy that causes signal revelation by the current agent when there are $x$ more $A$ actions than $B$ ones.
    Define distribution $\mathcal{D}$ as supported on $\mathcal{V}$, all possible values of $v_{x, p}\,,$ and having probability mass associated with any $v_{x, p}$ as $1/|\mathcal{V}|\,.$ 
    The government draws the subsidy they provide from this distribution $\mathcal{D}\,.$ 
\end{definition}

\subsection{Main Result}

In this section, we show that the random process described in Definition~\ref{defn:government-subsidy} applied to the setting described in Definition~\ref{defn:multigroup-game} for a reasonable amount of time shifts groups into correct cascades with high probability.

\begin{thm} \label{thm:kgroup}
    Suppose the government provides the subsidy detailed in Definition~\ref{defn:government-subsidy} for the game in Definition~\ref{defn:multigroup-game}. That is, at each time step $t > 0 \,,$ the government (without knowledge of the group of the current agent or the history) draws a subsidy at random $r_t \sim \mathcal{D}\,,$
    where the probability of any element in the support of $\mathcal{D}$ is at least $p_\text{min}\,.$ 
    Then, for all $\delta > 0,$ after $\frac{2\, k}{g_\text{min}} \, \left( \frac{ 2 \frac{\log(3k/\delta - 1)}{\log(p/(1-p))}}{\pmin \, (2p - 1)} + \frac{2 \log(3k/\delta)}{\pmin^2 (2p-1)^2}   + \log(3k/\delta)\right)$ steps, with probability at least $1-\delta\,,$ 
    all $k$ groups will end up in what is for them an up-cascade. 
\end{thm}

\textbf{Proof Sketch} At a high level, the proof proceeds as follows: (1) We fix a group and show that the subsidy behaves, as before, as a random walk that group takes on the number line. (2) We show how long it takes to achieve with high probability a sufficient condition for the walk to finish in a cascade on the correct action. (3) We union bound over the failure probability and appropriately scale the time required to achieve the stated result. 

\textbf{(1)} Let us first describe a random walk $\mathcal{R}_\text{gvt}$ that models the setting in Definition~\ref{defn:multigroup-game} with the government subsidy described in Definition~\ref{defn:government-subsidy}. See the appendix for a proof. Next, we define a related but simpler random walk, $\mathcal{R}$. In both cases, the reverse walk describes the same for a group for whom the correct direction is the left.

\begin{lemma} \label{lemma:random-walk-gvt}
    The following random walk, which we shall call $\mathcal{R}_\text{gvt}\,,$ describes the walk taken by a group for whom the correct action is to the right when the government provides the subsidy described in Definition~\ref{defn:government-subsidy}:
        $\Pr{i \text{ to } i-1} = \alpha_i \cdot (1-p) \, ; \,
        \Pr{i \text{ to } i} = 1-\alpha_i \, ; \,
        \Pr{i \text{ to } i+1} = \alpha_i\cdot p\,,$
    where $\alpha_i \coloneqq \mathbb{P}_\text{gvt choice}\left[\text{signal is revealed in this state}\right]\,.$ 
\end{lemma}

\begin{definition} \label{defn:rw-aug}
    Let us call the random walk with the following transition probabilities $\mathcal{R}$:
        $\Pr{i \text{ to } i-1} = p_\text{min} \cdot (1-p) \, ; \,
        \Pr{i \text{ to } i} = 1-p_\text{min} \, ; \,
        \Pr{i \text{ to } i+1} = p_\text{min}\cdot p\,.$ 
\end{definition}

In Lemma~\ref{lemma:analyse-rw-aug} in the appendix, we show that we can analyze $\mathcal{R}$ instead of $\mathcal{R}_\text{gvt}\,.$

\textbf{(2)} Next, we consider three modes of failure: first, if there are not sufficiently many signals aligned to the correct direction, the majority vote will not align with the correct action for the group; second, since the government only encourages signal-revealing but does not behave differently depending on which action is correct for a group, we may accidentally hit a bad cascade before hitting a good one just due to a bad ordering in the sequence of signals; third, we may not see enough people from this group to take sufficiently many steps on the random walk. We study the first in Lemma~\ref{lemma:whptrw}, plugging in a confidence parameter $\delta/(3k)$ for the failure probability. Then, the second is addressed by Lemma~\ref{lemma:prl-minusl}, again using $\delta/(3k)$ as the failure probability. Finally, we apply the Hoeffding bound with failure probability $\delta/(3k)$ to ensure we see sufficiently many people. With that, the analysis is complete for a single fixed group. 

\textbf{(3)} Finally, we union bound over the failure probability so that the result holds and report after how long of the government providing such a subsidy, with high probability {\em all} groups stabilize to optimism cascades. Detailed proofs for this theorem and its constituent lemmas are in the appendix. $\hfill \square$

\begin{rmk}
    For sake of generality, we present the result in Theorem~\ref{thm:kgroup} in terms of $\pmin\,,$ and $g_\text{min}\,.$ However, let us plug these in and discuss the scaling for intuition. 
    First, the support of the distribution $\mathcal{D}$ has size at most $2\Lrw$, so $p_\text{min} \le 1/(2L) \le 1/(2 \Lrwf) \,.$ Plugging this in, upper bounding, and ignoring constants, we get that the number of total required steps scales like $\frac{k\,L^2}{g_\text{min}} \, \cdot \, \frac{\log(3k/\delta)}{\left( p - \frac 12 \right)^2}\,,$ where $L = \frac{\log(3k/\delta - 1)}{\log(p/(1-p))}\,.$ We can see that the number of steps scales like $1/(p-1/2)^2\,,$ implying that the closer the signal strength to $1/2\,,$ the longer it takes to drift far enough in the walk. This is standard for problems where we must distinguish whether a ``coin flip'' has bias $1/2+\epsilon$ or $1/2-\epsilon\,.$ Next, we see the standard dependence on the failure probability, $\log(1/\delta)\,,$ and union bound, $\log 3k\,$. Finally, we see the inverse dependence on $g_\text{min}\,,$ meaning that the lower the minimum probability of seeing a group, the longer this subsidy needs to be in place. 
    Now, if the probability of an agent belonging to a group is uniform across the $k$ groups, then $g_\text{min} = 1/k\,,$ and so the scaling is like $ L^2 \, k^2 \, \cdot \, \frac{\log(3k/\delta)}{\left( p - \frac 12 \right)^2}\,.$ From this, we can see that the dominant dependence of this bound on the number of groups is through the prevalence of the lowest-prevalence group and making sure each group takes sufficiently many steps. 
\end{rmk}
 
\section{Experiments}

Finally, we conducted simulations to assess the effectiveness of financial supplements in overcoming pessimism traps in our theoretical model. We study the first setting, with a single group and correct action known to government. The primary objectives of our experiments were to evaluate the impact of financial supplements on the proportion of correct cascades, analyze the magnitude of the subsidies, and understand the scalability of interventions across population sizes.

\subsection{Data Generation and Procedure}

The data for our experiments were generated through simulations that model the sequential decision-making process of agents. Each agent must choose between two actions, $A$ and $B$, where $A$ is uniformly the better option, though this is unknown to the agents \textit{a priori}. The agents receive private signals indicating the correctness of their choice, with a probability $p$ of being correct. 
Each simulation was conducted as follows:
    We initialized a population of $N \in \{10, 100, 1000\}$ agents to explore the effects of population size on eventual cascade behavior.
    Each agent received a private signal with strength $p \in [0.51, 0.99]$.
    For each pair of values for $N$ and $p$, we repeat the experiment 100 times and average the results. For the proportion of correct cascades, we repeat each such experiment 10 times and report one standard deviation/$\sqrt{10}$ as the error bars. For the subsidy size, we report the standard deviation over the 100 trials / $\sqrt{100}$ as the error bars. 
    The agents sequentially made their decisions after observing the actions of all preceding agents.
    Financial supplements as derived above were introduced to influence the agents' decisions.

\subsection{Results}

\subsubsection{Impact of Financial Supplement on Correct Cascades}

\begin{figure}[ht!]
    \centering
    \includegraphics[width=0.8\textwidth]{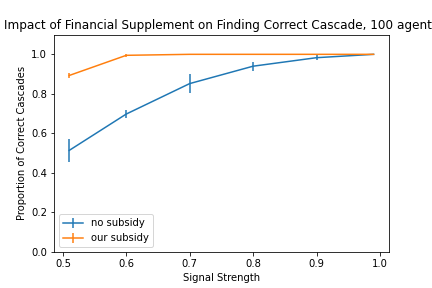}
    \caption{Impact of Financial Supplement on Finding Correct Cascade with Supplement, 100 agents.}
    \label{fig:using_supplement}
\end{figure}

\begin{figure}[ht!]
    \centering
    \includegraphics[width=0.8\textwidth]{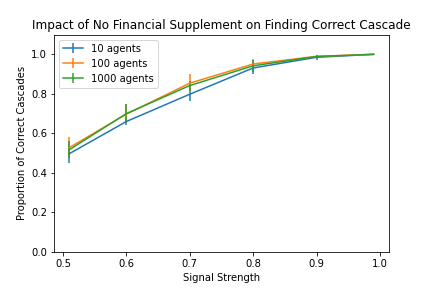}
    \caption{Probability of Finding Correct Cascade without Supplement \vspace{-5mm}}
    \label{fig:without_supplement}
\end{figure}

Figure \ref{fig:using_supplement} shows the proportion of correct cascades when a financial supplement is provided (orange line) to 100 agents compared to when it is not (blue line). The results indicate a significant improvement in the proportion of correct cascades, especially at lower signal strengths. For example, with a signal strength of 0.6, the proportion of correct cascades increases markedly (including across different population sizes (10, 100, and 1000 agents), see appendix for plots) when the supplement is used. This demonstrates the benefit of the proposed financial incentive in decision-making.

In contrast, Figure \ref{fig:without_supplement} shows the proportion of correct cascades without the supplement for varying numbers of agents. The performance is notably low, particularly for weaker signals. This highlights the promise of financial interventions in overcoming pessimism traps and steering agents toward optimal decisions. Also observe that the number of agents in the sequence has very little effect -- once a cascade forms, there is no new information to switch out of it even when there are new agents.

\subsubsection{Average Subsidy Progression}

To understand the dynamics of the financial supplement, we analyzed the progression of subsidies over time for different signal strengths. Figure \ref{fig:subsidy_100_agents} and Figures \ref{fig:subsidy_10_agents} and \ref{fig:subsidy_1000_agents} in the appendix depict these results for populations of 100, 10, and 1000 agents, respectively.

\begin{figure}[ht!]
    \centering
    \includegraphics[width=0.80\textwidth]{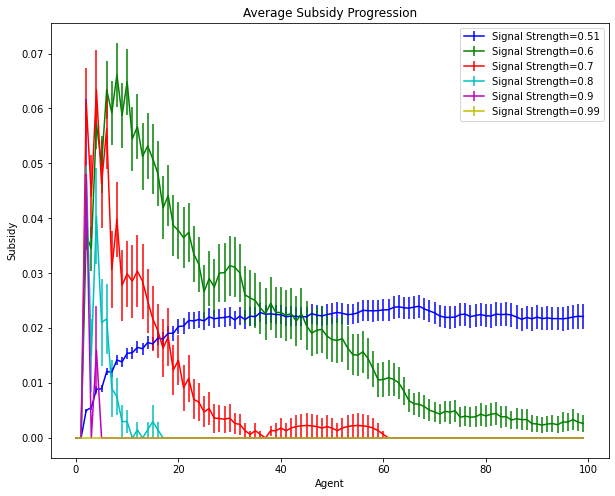}
    \caption{Average Subsidy Progression for 100 Agents \vspace{-3mm}}
    \label{fig:subsidy_100_agents}
\end{figure}

Figure \ref{fig:subsidy_100_agents} presents the average subsidy progression for a population of 100 agents. The required subsidy stabilizes after a few agents have made their decisions, indicating that once the initial agents are guided correctly, the need for subsequent subsidies diminishes, a benefit for allocation of resources. The line for the weakest signal levels off -- based on our theoretical results, we know that it takes a long time to stabilize when the signal margin is narrow.

Interestingly, the required supplement was smallest among the first few agents in the setting with the weakest private signal. This is easily seen in Figure \ref{fig:subsidy_10_agents}. In all three figures, we observe a spike in the supplement needed at the beginning, yet the necessary supplement typically returns to zero quite quickly. All three plots, however, show that in the case where $p=0.51$, the subsidy stays consistent across the population. This is likely because the signal is weak enough that it takes too long on average for the best action to be detected compared to the number of agents. It is also interesting to observe the average starting times for the subsidies, which align with the beginning of the pessimism trap and almost always begin within the first 10 rounds.

\subsection{Discussion and Implications}

The experiments validate our theoretical predictions regarding the impact of financial supplements on decision-making under uncertainty. The simulations provide several key insights. First, the introduction of financial supplements significantly improves the proportion of correct cascades, particularly in scenarios with lower signal strength and smaller populations. This underscores the importance of targeted interventions in guiding agents towards optimal decisions. Second, the progression of subsidies indicates that early interventions are critical. Once the initial agents are influenced correctly, the necessity for continued subsidies diminishes, suggesting an efficient use of resources. Lastly, the results suggest that the scalability of financial interventions is feasible. Larger populations benefit from reduced per-capita subsidies, making such policies more practical for broader applications.

\clearpage
\bibliography{refs}
\onecolumn

\appendix
\section{Key features of pessimism traps}
\label{appendix:features}
Morton's epistemic characterization of pessimism traps builds upon several empirical characterizations from scholars studying the phenomenon in the field of education. Morton elucidates pessimism traps via the following key distinguishing features [pgs. 732-3 of \cite{morton2022resisting}, paraphrased here]:
\begin{enumerate}
    \item Making the ambitious choice is an investment, involving a long duration before payoff and hard work.
    \item There is a feasible alternative that still has a reasonable payoff.
    \item There is risk involved in the ambitious choice, enough that it would affect the choice of a risk-averse agent.
    \item There is strong evidence that succeeding in the ambitious end when coming from the group of interest is of low likelihood.
    \item Not pursuing the ambitious end will not change the pessimistic view they hold about it.
\end{enumerate}

Importantly, the choices in question aren't even necessarily the ``right'' and ``wrong'' ones, but rather they are ``ambitious'' and ``moderate.'' 
\section{Proof for Posterior Update (Eqn.~\ref{eq:post})}

\begin{proof}
Given the conditional independence of $s_t$ and $\barH{t-1}$ when conditioned on the best action being either A or B, we can express their joint probability as:
\begin{align*}
\Pr{s_t, \barH{t-1}|\world{A}} &= \Pr{s_t|\world{A}} \Pr{\barH{t-1}|\world{A}}
\end{align*}

Then, applying Bayes' theorem, we have:
\begin{align*}
&\Pr{\world{A}|s_t, \barH{t-1}} = \frac{\Pr{s_t|\world{A}}\Pr{\barH{t-1}|\world{A}}\Pr{\world{A}}}{\Pr{s_t, \barH{t-1}}}
\end{align*}

Expanding the denominator using the law of total probability:
\begin{align*}
\Pr{s_t, \barH{t-1}} &= \Pr{\world{A}} \Pr{s_t, \barH{t-1}|\world{A}} + \Pr{\world{B}} \Pr{s_t, \barH{t-1}|\world{B}}
\end{align*}

Substituting the conditional independence:
\begin{align*}
\Pr{s_t, \barH{t-1}} &= \Pr{\world{A}}\Pr{s_t|\world{A}} \Pr{\barH{t-1}|\world{A}} + \Pr{\world{B}}\Pr{s_t|\world{B}} \Pr{\barH{t-1}|\world{B}}
\end{align*}

Finally, because we assume a uniform prior over which action is correct, we have $\Pr{\world{A}} = \Pr{\world{B}}$ and can thus cancel out all of those terms.

\end{proof}

\section{Formulation as a Random Walk} \label{appendix:random-walk}
Calculating the posterior can be difficult depending on the tie-breaking rule, as agents have to reverse engineer previous individuals' thought processes. However, with the tie-breaking rule we employ and our eventual subsidy, we can reformulate this process as a one-dimensional random walk on the integers. We begin with a toy example for intuition.

\begin{lemma}[Asymmetric Simple Random Walk in -1 to 1]
\label{lem:srw}
While the number of choices for action $A$ and action $B$ in the history $\barH{t-1}$ differ by no more than 1, agent $t$ will follow their signal. Thus, the process is an asymmetric simple random walk in this interval.
\end{lemma}

\begin{proof}
We will prove this through induction. Our inductive hypothesis is that if the previous $t-1$ agents acted according to their signal, and if the number of choices of action $A$ and $B$ differ by at most one, then agent $t$ will follow their signal. 

To begin, we establish the base case. Consider the first agent. Note that the hypothesis that the previous $t-1$ agents acted according to their signal is vacuously true. Now, let us look at the two signal options for the first agent to verify that they indeed follow their signal in each case. 

If the first agent receives signal $s_1=A$, they will update their posterior, according to Eqn.~\ref{eq:post} as:

\begin{align*}
\Pr{\world{A}|s_1=A, H_0 = \{\}} = \frac{\Pr{s_1=A|\world{A}}\Pr{H_0=\{\}|\world{A}}}{\sum_{w \in \{A,B\}}\Pr{s_1=A|\world{w}}\Pr{H_0=\{\}|\world{w}}}
\end{align*}

Now, because $H_0$ is always $\{\}$, we can set $\Pr{H_0=\{\}|\world{A}} = \Pr{H_0=\{\}|\world{B}} = 1$ and simplify the above expression to:

\begin{align*}
\Pr{\world{A}|s_1=A, H_0} &= \frac{\Pr{s_1=A|\world{A}}}{\Pr{s_1=A|\world{A}}+ \Pr{s_1=A|\world{B}}} = \frac{p}{p + (1-p)} = p
\end{align*}

Because $p > \frac{1}{2}$, the first agent will choose action $A$ in this case, thus following their signal.

If, instead the agent receives signal $B$, they will update their posterior as:
\begin{align*}
\Pr{\world{A}|s_1=B, H_0} &= \frac{\Pr{s_1=B|\world{A}}}{\Pr{s_1=B|\world{A}}+ \Pr{s_1=B|\world{B}}}
= \frac{1-p}{1-p + p}  = 1-p
\end{align*}

In this case, they will choose action $B$, again following their signal. From this, we can say that the first agent's action is identical to the action of their signal and so is distributed as a Bernoulli random variable with success parameter $p$.

 This satisfies the base case. Now, we proceed by induction. Suppose that the previous $t-1$ agents acted according to their signal, i.e., $H_{t-1} = \barH{t-1}$. Let $|A|$ indicate the number of choices for $A$ observed among the first $t-1$ agents, and let $|B|$ be the number of choices for $B$ among those agents. Then, if the $t^{th}$ agent gets signal $A$, their posterior is:

\begin{align*}
\Pr{\world{A}|s_t=A, H_{t-1}} = &\frac{p^{|A|+1}(1-p)^{|B|}}{p^{|A|+1}(1-p)^{|B|} + p^{|B|}(1-p)^{|A|+1}}
\end{align*}

If $|A| = |B|$, then the posterior is equal to $p$, and the agent follows their signal.
If $|A| = |B| - 1$, the posterior is equal to $\frac{1}{2}$, and the agent breaks the tie by following their signal. Finally, if $|A| = |B| + 1$, the posterior is $\frac{p^2}{p^2 + (1-p)^2}$, which is greater than $\frac{1}{2}$ for $p> \frac{1}{2}$.

Equivalently, if they get the signal $B$, the posterior is:

\begin{align*}
\Pr{\world{A}|s_t=B, H_{t-1}} = &
\frac{p^{|A|}(1-p)^{|B|+1}}{p^{|A|}(1-p)^{|B|+1} + p^{|B|+1}(1-p)^{|A|}}
\end{align*}

\end{proof}

\begin{lemma}[Stopping Points at -2 or 2]
\label{lem:stop}
A cascade will begin if the number of choices for action $A$ and action $B$ in the history $H_{t-1}$ differ by 2.
\end{lemma}

\begin{proof}
To see this, assume that at time $t$,  $|A| - |B| \geq 2$ or $|A| - |B| \leq -2$. Consider the first case. If agent $t$ gets signal $A$, their posterior is

\[
\frac{p}{1-p}^{|A|-|B|+1} \geq \frac{p}{1-p}^3 \geq 1 
\]

and so they take action $A$. If they get signal $B$, their posterior is  
\[
\frac{p}{1-p}^{|A|-|B|-1}\geq \frac{p}{1-p} \geq 1 
\]

So, regardless of their signal, they choose action $A$. Thus, a cascade begins, because this condition will be maintained for each subsequent agent. An equivalent analysis applies to the second case.
\end{proof}

\begin{lemma}[Probability of Wrong Cascade] \label{lemma:prcascade}
The probability of an incorrect cascade for $p>\frac{1}{2}$ is  $$\frac{\left( \frac{1-p}{p} \right)^2 + \left( \frac{1-p}{p} \right)^3}{1 + \left( \frac{1-p}{p} \right) + \left( \frac{1-p}{p} \right)^2 + \left( \frac{1-p}{p} \right)^3}\,.$$
\end{lemma}
\begin{proof}
We can compute this probability via a recurrence. Consider a random walk on the number line. Define $X_i$ as the probability a walk starting at $i$ reaches 2 before it reaches -2. Then, we have that:
\begin{align*}
    X_i = p \, X_{i+1} + (1-p) \, X_{i -1} \quad ; \quad 
    X_2 = 1 \quad ; \quad
    X_{-2} = 0
\end{align*}

We can solve this system as follows:
\begin{align*}
    Y_i &\coloneqq X_i - X_{i-1} \\
   p \, X_i + (1-p) \, X_i &=  p \, X_{i+1} + (1-p) \, X_{i -1} \Leftrightarrow (1-p) \, Y_i = p \, Y_{i+1} \\
   Y_i &= Y_0 \, \left(  \frac{1-p}{p} \right)^i \\
   \sum_{i = -1}^2 Y_i = X_2 - X_{-2} = 1 \quad \Rightarrow \quad
   Y_0  &= \frac{1}{\sum_{i = -1}^2 \left( \frac{1-p}{p} \right)^i} \\
   X_i &= X_{-2} + \sum_{j = -1}^i Y_j = \frac{1}{\sum_{j = -1}^2 \left( \frac{1-p}{p} \right)^j} \sum_{k = -1}^i \left(  \frac{1-p}{p} \right)^k \,.
\end{align*}
Simplifying this somewhat, we have that starting at 0:
\begin{align*}
    \Pr{\text{up cascade}} &= \frac{1 + \left( \frac{1-p}{p} \right)}{1 + \left( \frac{1-p}{p} \right) + \left( \frac{1-p}{p} \right)^2 + \left( \frac{1-p}{p} \right)^3} \\
    \Pr{\text{down cascade}} &= \frac{\left( \frac{1-p}{p} \right)^2 + \left( \frac{1-p}{p} \right)^3}{1 + \left( \frac{1-p}{p} \right) + \left( \frac{1-p}{p} \right)^2 + \left( \frac{1-p}{p} \right)^3} \\
\end{align*}

Note that as $(1-p)/p$ approaches 1 (i.e., $p \rightarrow 1/2)\,,$ the probability of a down cascade increases.
\end{proof}
Therefore, there is a substantial probability of an incorrect cascade, increasing as $p$ gets closer to $\frac{1}{2}$.

\begin{lemma}[Expected Starting Time] \label{lemma:estart} 
In expectation, it will take $\frac{2}{1-2p + 2p^2}$ agents for a cascade to begin.
\end{lemma}

\begin{proof}

To compute this, let us consider the same random walk defined in Lemma~\ref{lemma:prcascade}. With respect to this random walk, let us now define $Z_i = \E{\text{time to } \pm 2 \text{ starting from } i}\,.$ That is, we are interested in the expected time it takes to converge to a cascade. We can write the following equations and initial conditions and solve manually:

\begin{align*}
    Z_2 &= 0 \quad ; \quad
    Z_{-2} = 0 \\
    Z_i &= p\,(Z_{i+1} + 1) + (1-p)\, (Z_{i -1} + 1) = 1 + p\,(Z_{i+1}) + (1-p)\, (Z_{i -1})\\
    Z_0 &= 1 + p\,(Z_{1}) + (1-p)\, (Z_{-1}) \\
    Z_1 &= 1 + p \, Z_2 + (1-p)\, Z_0 = 1 + (1-p) Z_0 \\
    Z_{-1} &= 1 + p \, Z_0 + (1-p)\, Z_{-2} = 1 + p \, Z_0 \\
    Z_0 &= 1 + p\,(1 + (1-p) Z_0) + (1-p)\, (1 + p \, Z_0) = 2 + 2\, p \, (1-p) \, Z_0 = \frac{2}{1-2\,p\,(1-p)}\,.
\end{align*}

\end{proof}

\section{Time-Varying Subsidy Proofs}


\subsection{Proof of Theorem~\ref{thm:subsidy}}

\begin{proof}
Consider the scenario where the net expected reward for choosing action $A$ includes both the base reward $R$ and a potential subsidy $r_t$:
\begin{align*}
\left(R + r_t \right) \Pr{\world{A} \mid \barH{t-1}, s_t} + r_t \cdot \Pr{\world{B} \mid \barH{t-1}, s_t} = r_t + R \cdot \Pr{\world{A} \mid H_{t-1}, S_t}.
\end{align*}

Similarly, the reward for choosing action $B$ is:
\[
R \cdot \Pr{\world{B} \barH{t-1}, s_t}
\]

We seek to determine $r_t$ such that it ensures action $A$ is at least as preferable when $S_t = A$ and less preferable when $S_t = B$. This leads to the following conditions:
\begin{align*} \text{(1) } r_t + R \cdot \Pr{\world{A} \mid \barH{t-1}, s_t = A} &\geq  R \cdot \Pr{\world{B} \mid \barH{t-1}, s_t = A} \\ \text{(2) } r_t + R \cdot \Pr{\world{A} \mid \barH{t-1}, s_t = B} &\leq R \cdot \Pr{\world{B} \mid \barH{t-1}, s_t = B}
\end{align*}

Simplifying these inequalities, we find:
\begin{align*}
    r_t &\geq R \left( \Pr{\world{B} \mid \barH{t-1}, s_t = A}- \Pr{\world{A} \mid \barH{t-1}, s_t = A }\right), \\
    r_t &\leq R \left( \Pr{\world{B} \mid \barH{t-1}, s_t = B} - \Pr{\world{A} \mid \barH{t-1}, s_t = B}\right).
\end{align*}

Now, let us rewrite the probabilities by expanding the posterior in terms of the evidence (and recalling that the prior over best actions is uniform)
\begin{align*}
    \Pr{\world{A} \rvert \barH{t-1}, s_t = A} &= \frac{\Pr{\barH{t-1}, s_t = A \rvert \world{A}} }{\Pr{\barH{t-1}, s_t = A \rvert \world{A}} + \Pr{ \barH{t-1}, s_t = A \rvert \world{B}} } \\
    &= \frac{p \cdot \Pr{\barH{t-1} \rvert \world{A}} }{p \cdot \Pr{\barH{t-1} \rvert \world{A}}  + (1- p) \cdot \Pr{\barH{t-1} \rvert \world{B}}}
\end{align*}

\begin{align*}
    \Pr{\world{A} \rvert \barH{t-1}, s_t = B} &= \frac{\Pr{\barH{t-1}, s_t = B \rvert \world{A}}}{\Pr{\barH{t-1}, s_t = B \rvert \world{A}}+ \Pr{\barH{t-1}, S_t = B \rvert \world{B}} } \\
    &= \frac{(1-p) \cdot \Pr{\barH{t-1} \rvert \world{A}}}{p \cdot \Pr{\barH{t-1}\rvert \world{A}}  + (1- p) \cdot \Pr{\barH{t-1} \rvert \world{B}} }
\end{align*}

\begin{align*}
    \Pr{\world{B} \rvert \barH{t-1}, s_t = A} &= \frac{\Pr{\barH{t-1}, s_t = A \rvert \world{B}} }{\Pr{\barH{t-1}, s_t = A \rvert \world{A}} + \Pr{ \barH{t-1}, s_t = A \rvert \world{B}} } \\
    &= \frac{p \cdot \Pr{\barH{t-1} \rvert \world{B}} }{p \cdot \Pr{\barH{t-1} \rvert \world{A}}  + (1- p) \cdot \Pr{\barH{t-1} \rvert \world{B}}}
\end{align*}

\begin{align*}
    \Pr{\world{B} \rvert \barH{t-1}, s_t = B} &= \frac{\Pr{\barH{t-1}, s_t = B \rvert \world{B}}}{\Pr{\barH{t-1}, s_t = B \rvert \world{B}}+ \Pr{\barH{t-1}, S_t = B \rvert \world{B}} } \\
    &= \frac{(1-p) \cdot \Pr{\barH{t-1} \rvert \world{B}}}{p \cdot \Pr{\barH{t-1}\rvert \world{A}}  + (1- p) \cdot \Pr{\barH{t-1} \rvert \world{B}} }
\end{align*}

Suppose $|A|$ people have decided on $A$ in the history $\barH{t-1}$ (excluding any decisions made during a cascade). We define:
\begin{align*}
\gamma_t \coloneqq \frac{\Pr{\barH{t-1} \mid \world{B}}}{\Pr{\barH{t-1} \mid \world{A}}} = \frac{(1-p)^{|A|} p^{t-1-|A|}}{p^{|A|} (1-p)^{t-1-|A|}} = \left( \frac{1-p}{p} \right)^{2|A|-t} \frac{1-p}{p}.
\end{align*}

Then, we can re-write the above expressions as 
\begin{align*}
    \Pr{\world{A} \rvert \barH{t-1}, s_t = A} = \frac{1}{1  + \frac{1- p}{p} \gamma_t}
 \quad ; \quad
    \Pr{\world{B} \rvert \barH{t-1}, s_t = A} = \frac{\gamma_t }{1 + \frac{1- p}{p} \gamma_t}
\end{align*}

Similarly, we can re-write the expressions in our second condition as 
\begin{align*}
    \Pr{\world{A} \rvert \barH{t-1}, s_t = B} = \frac{1}{\frac{p}{1-p} + \gamma_t}
 \quad ; \quad
    \Pr{\world{B} \rvert \barH{t-1}, s_t = B} =  \frac{\gamma_t}{\frac{p}{1-p}  + \gamma_t}
\end{align*}

This implies that $R \left( \frac{\gamma_t -1 }{1 + \frac{1- p}{p} \gamma_t}\right) \leq r_t \leq R \left( \frac{\gamma_t - 1}{\frac{p}{1-p} + \gamma_t }\right)$

Finally, we show that there is always a non-negative value for the subsidy in this interval. Although this is not strictly required, we want to avoid a situation where an external entity is actually taking money from the agents. To do this, it is sufficient to see that $\gamma_t \geq 1$ in the region where the subsidy is needed to cause agents to act according to their signal and when the correct cascade has not yet been reached. By Lemmas \ref{lem:srw} and \ref{lem:stop}, this is the region where $|B| - |A| \geq 2$. Then,

\begin{align*}
\gamma_t = \frac{(1-p)^{|A|} p^{t-1-|A|}}{p^{|A|} (1-p)^{t-1-|A|}}  = \frac{(1-p)^{|A|} p^{|B|}}{p^{|A|} (1-p)^{|B|}} = \frac{p}{1-p}^{|B|-|A|}  \geq 1
\end{align*}

\end{proof}

\subsection{Proof of Theorem~\ref{thm:subsidy_budget}}
\label{appendix:proof-subsidylen}

\begin{proof}
We apply Wald's first equation to upper bound the amount of time for which the subsidy must be released. This time, we imagine that we must move from a state where there are two more choices for $B$ than $A$ to a state where there are two more choices for $A$ than $B$. This reduces to the problem of a random walk with probability $p$ of moving right hitting 2 when starting from -2 (with no left stopping point). Let $N$ be the stopping time at which this event occurs. Using Wald's equation, we have $\mathbb{E}[S_N]=4$ and $\mathbb{E}[X_1]= p \cdot 1 + (1-p) \cdot -1 = 2p-1$. Then,
$\mathbb{E}[N] = \frac{\mathbb{E}[S_N]}{\mathbb{E}[X_1]} = \frac{4}{2p-1}$. The corresponding subsidy calculation considers the worst-case scenario where each round necessitates the maximum subsidy $R$, leading to:
\begin{align*}
\mathbb{E}[\text{Total subsidy}] = R \times E(N) = R \times \frac{4}{2p-1}.
\end{align*}
\end{proof}

\section{Details for extension to multiple groups} \label{appendix:k-group}

\subsection{Defining The Setting}

We know that there is some way in which to subsidize the agent that incentivizes them to reveal their signal.
\begin{fact} \label{lemma:value-exists}
    For a group with signal strength $p$, for each value $x = |A|-|B|\,,$ there exists a value $v_{x, p}$ such that if the government provides a subsidy to the agent in play when there are $x$ more $A$ actions than $B$ actions of $v_{x, p}$ (i.e., subsidizes action $A$ by $v_{x, p}$ if $v_{x, p} > 0$ and action $B$ by $|v_{x, p}|$ otherwise), then the agent acts in such a way that their signal is revealed. The subsidy value $v_{x, p}$ is a function of $x\,,$ and $p\,,$ the probability that the signal is aligned with the correct action.
\end{fact}

\begin{proof}
    The existence of this subsidy is guaranteed by the fact that we construct such a subsidy in the previous section. In particular, see Algorithm~\ref{alg:subsidy} and Theorem~\ref{thm:subsidy}.
\end{proof}

\subsection{Analyzing the Subsidy Scheme}

We can model the state of a fixed group receiving a randomly chosen subsidy from the distribution $\mathcal{D}$ using a random walk, as before. In Lemma~\ref{lemma:random-walk-gvt}, we describe that random walk, following which in Definition~\ref{defn:rw-aug} we define a related walk. 

\subsubsection{Proof of Lemma~\ref{lemma:random-walk-gvt}}

\begin{proof}

    First, as before, we can formalize this process as a random walk, where the location on the random walk depends on the net difference between the number of people taking the ``positive'' action for their group and the number taking the ``negative'' action. It is clear that if an agent reveals their signal (which is ``right'' (both in terms of alignment with world and in direction on the walk) with probability $p$) they take a step in the ``right'' direction. Now, if the agent does not reveal their signal, there is no step taken on the random walk, since the location on the random walk is given by the number of revealed $A$ signals less the number of revealed $B$ signals. If the agent does not reveal their signal, then the action they take is ascribed by an agent to the subsidy they received and therefore does not affect the location on the random walk.
    We formalize this in the lemma before.
    
    \begin{lemma} \label{lemma:rw-iff-signal}
     A step of size 1 is taken on the random walk if and only if the agent revealed their signal. 
    \end{lemma}
    
    If (and only if, per the lemma above) the agent playing at a given time reveals their signal as a result of the government subsidy, they take a step on the random walk. 
    Since there could be a range of subsidies that are signal-revealing at a given location on the random walk for a value of $p\,,$ there may be multiple values the government could pick that would cause signal-revealing. Let these values be $\{ v_j \}_{j = 1}^m\,.$ Then, $\alpha_i \coloneqq \sum_{j} \Pr{v_j}\,.$ Now, with probability $\alpha_i\,,$ the signal is revealed and a step is taken, and with probability $1-\alpha_i\,,$ the signal is not revealed and the walk stays in the same state. If the signal is revealed, then as before, the walk proceeds from $i$ to $i+1$ with probability $p$ and from $i$ to $i-1$ with probability $1-p\,.$
    
\end{proof}

\begin{rmk}
    If a group for whom the correct action is A / right reaches $+L > 2\,,$ then they are in the correct cascade for them. Similarly, if B / left is the correct action for a group and they reach $-L < -2$ in the random walk, then they are in the correct cascade for them.
\end{rmk}

    We showed in the previous section in Lemma~\ref{lem:stop} that $L=2$ actually suffices for pushing people into the cascade. However, in this setting, because the government provides the signal-revealing subsidy even after one group reaches $L=2\,,$ it may be possible to get a ``bad'' string of signals that push a group back toward the pessimism trap. Thus, we need to sustain it for longer. We show in subsequent lemmas exactly for how much longer we need to apply it.

Next, we analyze the simplified random walk $\mathcal{R}$ (Lemmas~\ref{lemma:prl-minusl}, \ref{lemma:whptrw}) and then show that we can use the analysis of $\mathcal{R}$ to study $\mathcal{R}_\text{gvt}$ (Lemma~\ref{lemma:analyse-rw-aug}).

\begin{lemma} \label{lemma:prl-minusl}
    In random walk $\mathcal{R}$ as defined in Definition~\ref{defn:rw-aug}, the probability that when starting at location $i\,,$ the walk hits $\Lrw$ before hitting $-\Lrw$ is given by:
    \begin{align*}
    x_i &= \Pr{\text{walk hits } \Lrw \text{ before hitting } -\Lrw \text{ when starting at } i}  \\
    & = \frac{2p-1}{(1-p)\, \left( \alpp^{-L} - \alpp^{L}  \right)} \sum_{j = -\Lrw + 1}^{i} \alpp^j\,.
    \end{align*}
    Thus, in order for $x_0 \ge 1-\delta\,,$ i.e., for the probability of hitting the correct cascade before the incorrect one when the walk starts from 0 to be high, we need $L = \frac{\log(1/\delta - 1)}{\log(p/(1-p))}\,.$ 
\end{lemma}

\begin{proof}
    We can analyze this by means of a recurrence. For $x_i$ defined as in the lemma statement, we have that:
    \begin{align*}
        x_i &= (1-p_\text{min}) \, x_i + p_\text{min} \, p \, x_{i + 1} + p_\text{min} \, (1-p) \, x_{i - 1} \\
        \Leftrightarrow x_i &=  p \, x_{i + 1} +  (1-p) \, x_{i - 1} \\
        \Leftrightarrow p \, ( x_i - x_{i + 1}) &= (1-p) \, (x_{i -1} - x_{i}) \\
    \end{align*}
    Now, we write down the boundary conditions: $x_{-\Lrw} = 0$ and $x_{\Lrw} = 1\,.$ Next, we solve the dynamics:
    \begin{align}
        \text{Define } y_i \coloneqq x_{i} - x_{i-1}
        &\Rightarrow (1-p) \, y_{i} = p \, y_{i+1}\\
        \text{Solution } \qquad \quad \qquad y_i &= c \left(\frac{1-p}{p}\right)^i \\
        x_\Lrw - x_{-\Lrw} = 1 &= \sum_{i = -\Lrw + 1}^{\Lrw} y_i = \sum_{i = -\Lrw + 1}^{\Lrw} c \alpp^i \\
        \Leftrightarrow c &= \frac{1}{\sum_{i = -\Lrw + 1}^{\Lrw} \alpp^i} = \frac{1 - \frac{1-p}{p}}{\alpp^{-L + 1} \left(1 - \alpp^{2L + 2} \right) } \\
        &= \frac{\frac{2p-1}{p}}{\alpp\, \left( \alpp^{-L} - \alpp^{L}  \right)}
    \end{align}

    Plugging this back into the definition of $y_i\,,$ we get:
    \begin{align}
        x_i = x_{i - 1} + \frac{2p-1}{(1-p)\, \left( \alpp^{-L} - \alpp^{L}  \right)} \, \alpp^i
    \end{align}

    Finally, let us incorporate the initial condition $x_{-\Lrw} = 0\,.$ Then, $$x_{-\Lrw + 1} = \frac{2p-1}{(1-p)\, \left( \alpp^{-L} - \alpp^{L}  \right)} \, \alpp^{-\Lrw + 1}\,,$$ and in general:
    $$
    x_i = \frac{2p-1}{(1-p)\, \left( \alpp^{-\Lrw} - \alpp^{\Lrw}  \right)} \sum_{j = -\Lrw + 1}^{i} \alpp^j\,.
    $$
    Note that when $ i = \Lrw\,,$ indeed $x_\Lrw = 1\,.$

    Finally, we compute $x_0 = \frac{2p-1}{(1-p)\, \left( \alpp^{-\Lrw} - \alpp^{\Lrw}  \right)} \sum_{j = -\Lrw + 1}^{0} \alpp^j\,$ and set this to be at least $1-\delta$ so we can determine what the appropriate threshold $L$ is. In particular:
    \begin{align}
        x_0 &= \frac{2p-1}{(1-p)\, \left( \alpp^{-\Lrw} - \alpp^{\Lrw}  \right)} \left( \frac{\alpp^{-\Lrw + 1} - \alpp}{1 - \alpp}  \right) \\
        &= \frac{\alpp^{-\Lrw} - 1}{\alpp^{-\Lrw} - \alpp^\Lrw} = \frac{1}{1 + \alpp^{\Lrw}} \\
        \text{want } \quad &\ge 1-\delta\,.
    \end{align}
    Solving, we get $L \ge \frac{\log(1/\delta - 1)}{\log(p/(1-p))}\,.$

\end{proof}

\begin{lemma} \label{lemma:whptrw}
    After $T = \frac{2 L}{\pmin \, (2p - 1)} + \frac{2 \log(1/\delta)}{\pmin^2 \, (2p-1)^2}$ steps of this random walk, with probability at least $1-\delta$, the walk will have moved at least $L\,$ steps net to the right.
\end{lemma}

\begin{proof}
    Let us define the outcome of step $i$ in the following way:
    $$
    X_i \coloneqq \begin{cases} +1 & \text{ step right, i.e., w.p. } \, p_\text{min} \, p \\
    0 & \text{ step right, i.e., w.p. } \, 1-p_\text{min} \\
    -1 & \text{ step right, i.e., w.p. } \, p_\text{min} \, (1-p)   \end{cases}\,.
    $$
    We are interested in upper bounding $\Pr{\sum_{i = 1}^T X_i \le \Lrw}\,.$ To do so, we will apply a Hoeffding bound. In particular, we have that $\mathbb{E}\left[ X_i  \right] = - \pmin + \pmin \, p + \pmin \, p = \pmin (2 \, p - 1)\,$ and so $\mathbb{E}\left[ \sum_{i = 1}^T X_i  \right] = \pmin \, T \, (2\, p - 1)\,.$ Applying the Hoeffding bound:
    \begin{align}
        \Pr{\sum_{i = 1}^T X_i \le \mathbb{E}\left[ \sum_{i = 1}^T X_i  \right] - \left(\mathbb{E}\left[ \sum_{i = 1}^T X_i  \right] - \Lrw \right)} &\le \exp \left( \frac{-2 \left( \Lrw -  \pmin \, T \, (2\, p - 1)  \right)^2}{T \cdot 4}    \right) \\
        &= \exp \left( \frac{-2 \left( \Lrw -  \pmin \, T \, (2\, p - 1)  \right)^2}{4 T}    \right)\,.
    \end{align}
    Finally, we can solve the following quadratic in $T$ to show that the stated value of $T$ suffices:
    \begin{align}
            4 \log(1/\delta)\,T &\le 2 \Lrw^2 - 4 \Lrw \, \pmin \, (2p -1) \, T \,+ 2\,\pmin^2\, T^2 \, (2p - 1)^2 \\
            T &\ge  \max \bigg \{   \frac{-(- 4 \Lrw \, \pmin \, (2p -1) - 4 \log(1/\delta))}{4\,\pmin^2\,(2p - 1)^2} \pm \\
            & \qquad \frac{\sqrt{(- 4 \Lrw \, \pmin \, (2p -1) - 4 \log(1/\delta))^2 - 4(2\,\pmin^2\,(2p - 1)^2)(2 \Lrw^2)}}{4\,\pmin^2\,(2p - 1)^2}   \bigg \} \\
            \Leftarrow T &\ge \frac{2 \left( 4 \Lrw \, \pmin \, (2p -1) + 4 \log(1/\delta) \right) }{4\,\pmin^2\,(2p - 1)^2}   \\
            &= \frac{2L}{\pmin \, (2p - 1)} + \frac{2 \log(1/\delta)}{\pmin^2 \, (2p-1)^2}
    \end{align}

\end{proof}

We show that we can analyze $\mathcal{R}_\text{gvt}$ by analyzing the simpler-to-analyze $\mathcal{R}\,.$

\begin{lemma} \label{lemma:analyse-rw-aug}
    We can analyze the random walk $\mathcal{R}_\text{gvt}$ described in Lemma~\ref{lemma:random-walk-gvt}, i.e., the one reflecting the dynamics of the game in the presence of the described government subsidy, by instead analyzing the random walk $\mathcal{R}\,$ given in Definition~\ref{defn:rw-aug}. In particular, the probability $x_i = \Pr{\text{walk hits } \Lrw \text{ before hitting } -\Lrw \text{ when starting at } i}$ is the same in both random walks,  $\mathcal{R}\,, \mathcal{R}_\text{gvt}\,.$ Further, let $T_\mathcal{R}$ be the number of steps needed to be taken in $\mathcal{R}$ for the walk to hit $L\,$ with high probability. Suppose $T_{\mathcal{R}_\text{gvt}}$ is the number of steps needed to achieve the same in $\mathcal{R}_\text{gvt}$. Then, $T_{\mathcal{R}_\text{gvt}} \le T_\mathcal{R}\,.$
\end{lemma}

\begin{proof}
    We consider the two parts separately.

    \textbf{Probability} Let us write the recurrence equations for each of the random walks. First, for $\mathcal{R}_{\text{gvt}}$:
    $$
    x_i = (1-\alpha_i)\, x_i + \alpha_i\, (1-p)\, x_{i - 1} + \alpha_i\, p\, x_{i + 1}\,.
    $$
    And for $\mathcal{R}\,:$
    $$
    x_i = (1-p_\text{min})\, x_i + p_\text{min} \,(1-p)\, x_{i - 1} + p_\text{min}\, p\, x_{i + 1}\,.
    $$
    Now, observe that subtracting the first term on the right hand side from both sides we get, respectively:
    \begin{align*}
        \alpha_i\, x_i &= \alpha_i\, (1-p) \,x_{i - 1} + \alpha_i\, p\, x_{i + 1} \\
        p_\text{min}\, x_i &= p_\text{min}\, (1-p)\, x_{i - 1} + p_\text{min}\, p\, x_{i + 1}\,.
    \end{align*}
    This can be simplified to:
    \begin{align*}
        x_i &= (1-p)\, x_{i - 1} +  p \, x_{i + 1} \\
         x_i &= (1-p)\, x_{i - 1} + p \, x_{i + 1}\,.
    \end{align*}
    These equations are the same, meaning that the dynamics are the same. If the boundary conditions are the same, then the values will also be the same.

    \textbf{Time Required To Converge} For this, we simply observe that the relationship between $\alpha_i$ in $\mathcal{R}_\text{gvt}$ and $\pmin$ in $\mathcal{R}$ is $\alpha_i \ge \pmin \, \forall i\,.$ Let us consider the instances where a step is taken in the original walk $\mathcal{R}_\text{gvt}$ as compared to where steps are taken in the analyzed walk $\mathcal{R}\,.$ Define $D_i \coloneqq \mathbb{I}\left[ \text{ step taken in } \mathcal{R}  \right]\,.$
    Define:
    $$
    \tilde{D}_i = \begin{cases}
    1 & \text{if }\,D_i = 1 \\
    1 & \text{ w.p. } \frac{\alpha_i - \pmin}{1-\pmin} \text{ if } \, D_i = 0 \\ 
    0 & \text{ otherwise.}
    \end{cases}
    $$
    
    Note that $\tilde{D}_i$ is $1$ with probability $\alpha_i\,$ and 0 with probability $1-\alpha_i\,.$ 
    Additionally, due to the coupling, $T_\alpha \coloneqq \sum_{i = 1}^T \tilde{D}_i \ge \sum_{i = 1}^T D_i \eqqcolon T_{\pmin} \,.$ 
    As a result, there are strictly more steps taken in $\mathcal{R}_\text{gvt}$ than in the $\mathcal{R}\,.$ Conditioned on taking a step, the probabilities of right and left steps are the same.  
    Thus, for a threshold of interest $Q\,,$ the probability of not exceeding it after $T_\alpha$ steps is smaller than the probability of not exceeding it after $T_{\pmin}$ steps.
   
\end{proof}

\subsubsection{Proof of Theorem~\ref{thm:kgroup}}
Finally, we present the proof of the main theorem, Theorem~\ref{thm:kgroup}.

\begin{proof}

\newcommand{\green}[1]{\textcolor{green}{#1}}

We start by fixing a group to look at. We show that the subsidy behaves, as before, as a random walk that group takes on the number line and show how long it takes to achieve with high probability a sufficient condition for the walk to finish in an up cascade. We then union bound over the failure probability and appropriately scale the time required to get the stated result.

Thus, let us start by fixing a group. Suppose for this group that $A$ is the correct choice (``right'' is the direction for an up cascade). All of our subseequent arguments hold symmetrically in the case that $B$ is the correct choice (``left'' is the direction for an up cascade). By Lemma~\ref{lemma:random-walk-gvt}, we have the description of a random walk $\mathcal{R}_\text{gvt}$ that models the setting in Definition~\ref{defn:multigroup-game} with the government subsidy described in Definition~\ref{defn:government-subsidy}. Further, Lemma~\ref{lemma:analyse-rw-aug} establishes that for the quantities of our interest, we can study $\mathcal{R}$ defined in Definition~\ref{defn:rw-aug} instead of $\mathcal{R}_\text{gvt}\,.$ 

With this having been established, let us define a tolerance parameter $\delta_0 \coloneqq \delta/(3k)$ which represents the probability of a single failure for a single group. If there is probability at most $\delta_0$ that the walk does not net $L$ steps and probability at most $\delta_0$ that even if the walk nets at least $L$ steps, it hit $-L$ first, then the total failure probability for this walk is at most $2\delta_0 = \delta/k\,.$ We want the group to take sufficiently many steps in the random walk that it exceeds a threshold $L\,,$ sending it into an up cascade. In the past, we considered $L = 2$ because the government knew the correct action and only subsidized it, but now that we're subsidizing both actions, we need to protect against hitting a bad cascade before hitting the correct one. By Lemma~\ref{lemma:whptrw}, we have that $T = \frac{2L}{\pmin \, (2p - 1)} + \frac{2 \log(1/\delta_0)}{\pmin^2 (2p-1)^2}\,$ steps suffice for netting $L$ steps to the right with probability at least $1-\delta_0 = 1-\delta/(3k)\,.$ 

Now, we do want to make sure the walk did not hit a down cascade before it hits the up cascade. For this, we analyze the probability of hitting the down cascade first in Lemma~\ref{lemma:prl-minusl}. As before, we want the probability of failure to be at most $\delta_0\,,$ and so applying Lemma~\ref{lemma:prl-minusl}, $L \ge \frac{\log(1/\delta_0 - 1)}{\log(p/(1-p))}\,.$

\newcommand{\nagentsneed}{\frac{ 2 \frac{\log(1/\delta_0 - 1)}{\log(p/(1-p))}}{\pmin \, (2p - 1)} + \frac{2 \log(1/\delta_0)}{\pmin^2 (2p-1)^2}}

Thus, with probability at least $1-2\delta_0 = 1-2\delta/(3k)\,,$ this group will end in an up cascade after $\nagentsneed$ steps (i.e., agents from this group are seen) without any negative side effects. 

Next, we analyze how long it takes to see enough agents from this group with high probability. An agent is from group $j$ with probability at least $g_\text{min}\,,$ and so with probability at least $1-\delta_0\,,$ after $\frac{2T^\star + 2\log(1/\delta_0)}{g_\text{min}}$ agents, we have seen agents from that group at least $T^\star$ times. Thus, we have that with probability at least $1-\delta/k\,,$ after $\frac{2\left( \nagentsneed \right) + 2\log(3k/\delta)}{g_\text{min}}$ steps, with high probability, for a fixed group, we have (1) seen enough agents from that group to (2) see enough ``correct'' signals for that group and (3) not enter a bad cascade first. Finally, we union bound over the $k$ groups to get that with probability at least $1-\delta\,,$ $k$ times this number of steps suffices for {\em all} groups to reach positive cascades.

\end{proof}

\onecolumn
\section{Extra Experiments}

In this section, we provide additional plots from our experiments. 

\subsection{Proportion Correct Cascades}
\begin{figure}[ht!]
    \centering
    \includegraphics[width=0.75\textwidth]{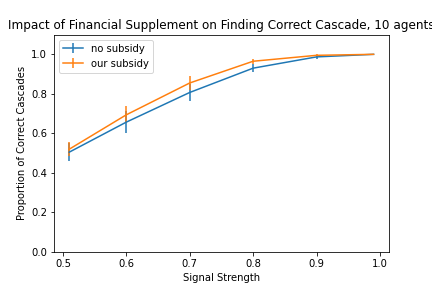}
    \caption{Proportion Correct Cascades for 10 Agents}
    \label{fig:subsidy_10_agents}
\end{figure}

\begin{figure}[ht!]
    \centering
    \includegraphics[width=0.75\textwidth]{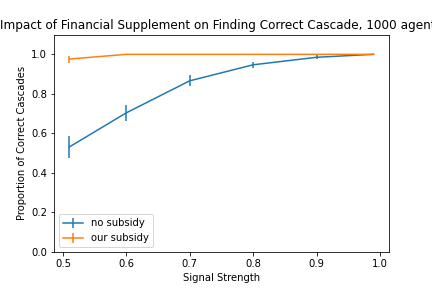}
    \caption{Proportion Correct Cascades for 1000 Agents}
    \label{fig:subsidy_1000_agents}
\end{figure}

\subsection{Subsidy Size}

\begin{figure}[ht!]
    \centering
    \includegraphics[width=0.75\textwidth]{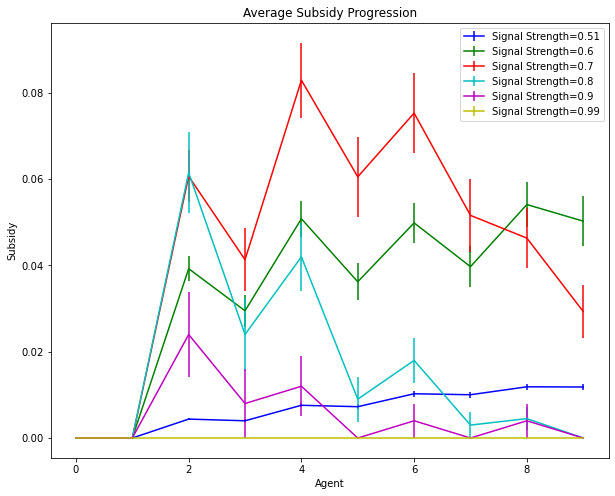}
    \caption{Average Subsidy Progression for 10 Agents}
    \label{fig:subsidy_10_agents_prog}
\end{figure}

\begin{figure}[ht!]
    \centering
    \includegraphics[width=0.75\textwidth]{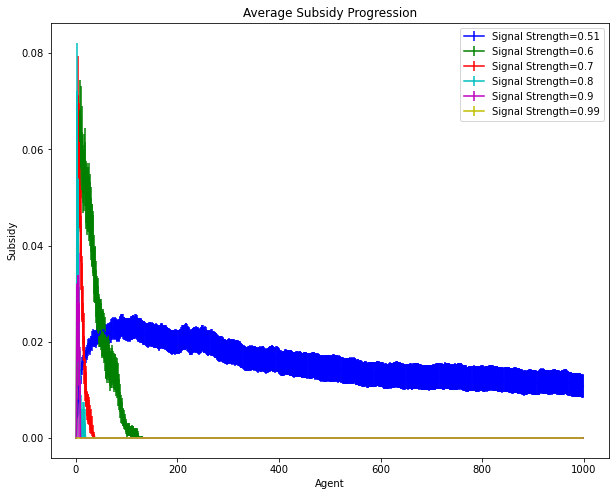}
    \caption{Average Subsidy Progression for 1000 Agents}
    \label{fig:subsidy_1000_agents_prog}
\end{figure}

\end{document}